\documentclass[10pt, conference, letterpaper]{IEEEtran}
\IEEEoverridecommandlockouts
\usepackage{cite}
\usepackage{amsmath,amssymb,amsfonts}
\usepackage{graphicx}
\usepackage{textcomp}
\usepackage{xcolor}

\usepackage{subfig}
\usepackage{booktabs,diagbox}
\usepackage{multirow}

\usepackage[linesnumbered,ruled,vlined]{algorithm2e}
\makeatletter
\newcommand{\removelatexerror}{\let\@latex@error\@gobble}
\makeatother

\newtheorem{theorem}{Theorem}
\newtheorem{proof}{Proof}

\def\BibTeX{{\rm B\kern-.05em{\sc i\kern-.025em b}\kern-.08em
    T\kern-.1667em\lower.7ex\hbox{E}\kern-.125emX}}
\begin{document}

\title{OTAS: An Elastic Transformer Serving System via Token Adaptation\\
\thanks{
The work described in this paper was substantially supported by two grant from the Research Grants Council of the Hong Kong Special Administrative Region, China (Project No. PolyU15222621, PolyU15225023), the NSF of China (No. 62172375, No. 62302184), Provincal Key Research and Development Program of Hubei (No. 2023BAB065),  the Key-Area Research and Development Program of Guangdong Province (No. 2021B0101400003), Hong Kong RGC Research Impact Fund (No. R5060-19, No. R5034-18), Areas of Excellence Scheme (AoE/E-601/22-R), and General Research Fund (No. 152203/20E, 152244/21E, 152169/22E, 152228/23E).}
}

\author{\IEEEauthorblockN{Jinyu Chen$^{1}$, Wenchao Xu*$^{1}$, Zicong Hong*$^{1}$, Song Guo$^{2}$, Haozhao Wang$^{3}$, Jie Zhang$^{1}$, and Deze Zeng$^{4}$}
\IEEEauthorblockA{$^1$Department of Computing, The Hong Kong Polytechnic University\\
$^2$Department of Computer Science and Engineering, The Hong Kong University of Science and Technology\\
$^3$School of Computer Science and Technology, Huazhong University of Science and Technology\\
$^4$School of Computer Science, China University of Geosciences\\
jinyu.chen@connect.polyu.hk, wenchao.xu@polyu.edu.hk, zicong.hong@connect.polyu.hk\\ songguo@cse.ust.hk, hz\_wang@hust.edu.cn, jie-comp.zhang@polyu.edu.hk, deze@cug.edu.cn\\
*Corresponding author: Wenchao Xu and Zicong Hong}
}


\maketitle

\begin{abstract}
Transformer model empowered architectures have become a pillar of cloud services that keeps reshaping our society. However, the dynamic query loads and heterogeneous user requirements severely challenge current transformer serving systems, which rely on pre-training multiple variants of a foundation model, i.e., with different sizes, to accommodate varying service demands. Unfortunately, such a mechanism is unsuitable for large transformer models due to the additional training costs and excessive I/O delay. In this paper, we introduce \textsc{OTAS}, the first elastic serving system specially tailored for transformer models by exploring lightweight token management. We develop a novel idea called \emph{token adaptation} that adds prompting tokens to improve accuracy and removes redundant tokens to accelerate inference. To cope with fluctuating query loads and diverse user requests, we enhance \textsc{OTAS} with application-aware selective batching and online token adaptation. \textsc{OTAS} first batches incoming queries with similar service-level objectives to improve the ingress throughput. Then, to strike a tradeoff between the overhead of token increment and the potentials for accuracy improvement, \textsc{OTAS} adaptively adjusts the token execution strategy by solving an optimization problem. We implement and evaluate a prototype of \textsc{OTAS} with multiple datasets, which show that \textsc{OTAS} improves the system utility by at least 18.2\%.
\end{abstract}


\section{Introduction}
It is of vital importance for the cloud to effectively serve the machine learning models for artificial intelligence (AI) applications, which can substantially affect the 
quality of user experience and the accompanied economic profits. For example, Facebook has 1.82 billion daily active users and issues tens of trillions of model inference queries per day, which necessitates fundamental re-designs for facilitating the model optimization and the serving efficiency~\cite{hpc-ai}.

Recent advances in self-supervised pre-training techniques have boosted the development of large transformer models while imposing a substantial burden on the model serving. These pre-trained transformer models have dramatically revolutionized our lives and brought remarkable potential to our society. For example, large pre-trained models like GPT-3 have spawned a host of applications, such as Copilot~\cite{copilot} and ChatGPT~\cite{ChatGPT}. In particular, ChatGPT has more than 100 million active users and received 176 million visits in April 2023~\cite{Similarweb}. 
Despite the explosion of such diverse applications, the resource-intensive nature of transformer models,  coupled with dynamic query loads and heterogeneous user requirements have exacerbated the challenges associated with transformer serving, making it extremely challenging to accommodate various service demands. In this context, the implementation of an elastic serving system that can adapt the serving process for improving service quality emerges as a promising solution. 


    \begin{figure}[t]
    \centering
	\subfloat[\emph{Model Adaptation.} The service provider prepares multiple versions of transformers, e.g., tiny (T), small (S), base (B) and large (L) models, and switches them during runtime.]
        {\includegraphics[width = 0.23\textwidth]{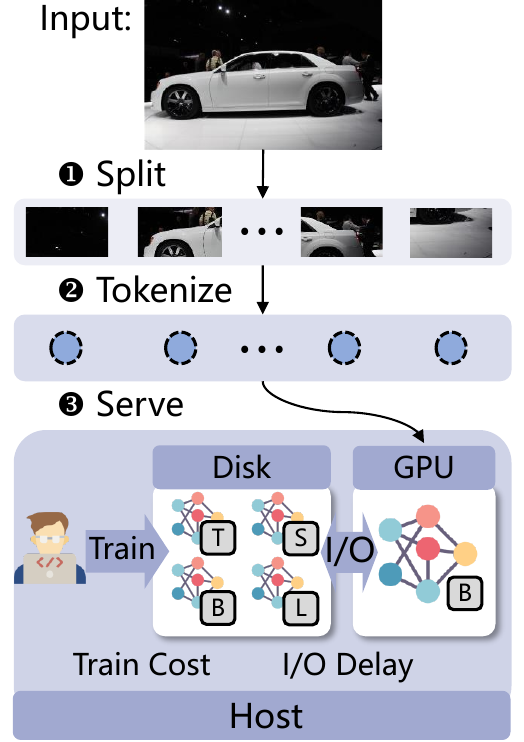}\label{fig:model_introduction}}
	\hfill
	\subfloat[\emph{Token Adaptation.} We realize elastic transformer serving through lightweight token adaptation that adds or reduces execution tokens according to service characteristics.]
        {\includegraphics[width = 0.23\textwidth]{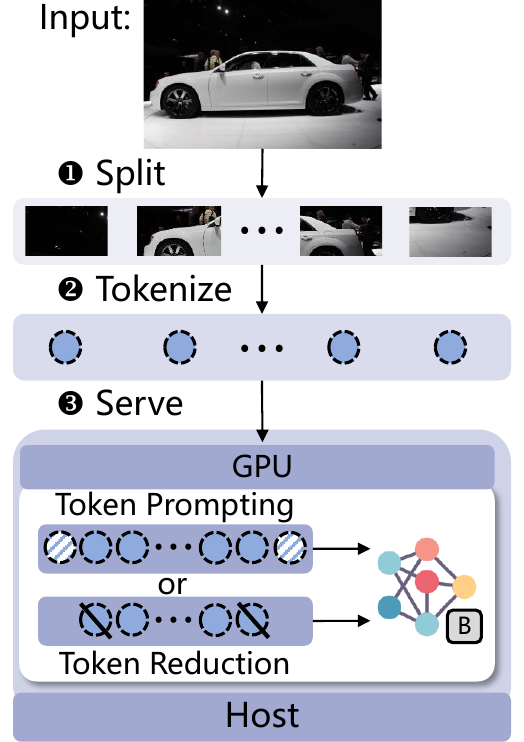}\label{fig:token_introduction}}
    \caption{Comparison between model adaptation and token adaption.}
    \label{fig:introduction_2}
    \end{figure}

As shown in Fig.~\ref{fig:model_introduction}, a common approach to realize an elastic serving platform, called \emph{model adaptation}, is to pre-train multiple model variants and dynamically select one to accommodate the variations in query load~\cite{10.14778/3570690.3570692,wang2023tabi,10.1145/3575693.3575698,romero2021infaas}. 
Unfortunately, it is unsuitable for large transformer models because loading such kind of model to GPU may creates prohibitive I/O overhead~\cite{sheng2023high}.
This scheme also leads to additional training costs, such as high monetary costs and time overhead.
Moreover, the size of different transformer models varies significantly, and it is hard to prepare different fine-grained model versions.
In this paper, we delve into the inherent characteristic of transformer models and develop a novel idea called \emph{token adaptation} for elastic model serving. 
As shown in Fig.~\ref{fig:token_introduction}, a token is a basic unit of text, code, or a patch of an image~\cite{microtoken}. 
The transformer models process these tokens with attention mechanism~\cite{vaswani2017attention}, which calculates the similarity (i.e., attention weight) between the query and key, and projects the attention value with the similarity weight to get a new representation.
One of the key properties of attention is its ability to support token sequences of varying lengths. 
A line of works have shown that input tokens have a large influence on the model performance, such as the accuracy and running time.
For example, identifying and removing redundant or unnecessary tokens, such as those representing the background, can expedite inference with little accuracy drop~\cite{bolya2022tome}.
On the other side, recent studies also demonstrate that adding prompting tokens that contain specific semantic features of the object and interact with input tokens can help produce more accurate results~\cite{jia2022visual}. 
%
Motivated by the above findings, we seek to explore a novel design space about token adaptation that dynamically adjusts the execution tokens of the transformer model to improve service quality with negligible training and I/O costs, i.e., serving more important queries. 



Despite the promising potential, designing token adaptation for a serving system is non-trivial and faces the following challenges:  (1) \emph{Diverse user requests.} The queries in a batch can vary significantly in query content, task, utility reward, and latency requirements. 
The allocation of token numbers needs to balance the demands of different requests.
(2) \emph{Fluctuating query load.} In real scenarios, the query load is fluctuating and bursty. The system should accommodate different numbers of requests with a short reaction time. 
(3) \emph{Model design.} None of the existing transformer models can support fine-grained token management. A unified transformer model needs to process prompt parameters from various tasks.

Therefore, in this paper, we design an elastic transformer serving \textbf{S}ystem via \textbf{O}nline \textbf{T}oken \textbf{A}daptation, named \textsc{OTAS}.
We first introduce the modules of the system and present the pipeline to process the incoming queries.  
Then, we propose a unified transformer model that can flexibly adjust the execution schemes with token prompting and token reduction. 
Moreover, we present an adaptive batching algorithm to group queries with similar service-level objectives, such as latency requirement and utility value, to improve throughput.
To cope with fluctuating query loads and diverse user requirements, we formulate a utility maximization problem with user latency constraints to adaptively adjust the token execution strategy.
Then, we design an efficient dynamic programming algorithm to derive the token execution plan for a batch.

We summarize the contributions as follows.
\vspace{-4pt}
\begin{itemize}
    \item We explore a novel design space for transformer serving, called token adaptation, allowing for flexible manipulation of the token execution plan for transformers. It reveals a new trade-off space between the potential for improving accuracy and the cost of inference latency.
    \item We present \textsc{OTAS}, an elastic transformer serving system via an online token allocation algorithm. We design a batching algorithm that groups similar queries for execution to improve throughput. 
    Besides, we adaptively allocate the token number for a batch to cope with the dynamic query load and diverse user requirements.
    \item We implement a prototype system of \textsc{OTAS}. The experimental results show that \textsc{OTAS} improves the serving utility by at least 18.2\% and serves more requests.
\end{itemize}

\vspace{-6pt}
\section{Background \& Motivation}
\vspace{-2pt}
\subsection{Transformer Model}
    
    Currently, transformer models, which show superior abilities on various vision and language tasks, becomes the mainstream backbone for neural network~\cite{dosovitskiy2020image,devlin2018bert,fedus2022switch}. The transformer model is notable for its ability to capture long-range dependencies for sequential data and be easily scaled up to millions or even billions of parameters~\cite{brown2020language,dehghani2023scaling}. The current paradigm of AI development has been switched to adopt a large-scale pre-trained transformer model for serving~\cite{liu2023pre}.

    \begin{figure}[t]
      \centering
      \includegraphics[width=\linewidth]{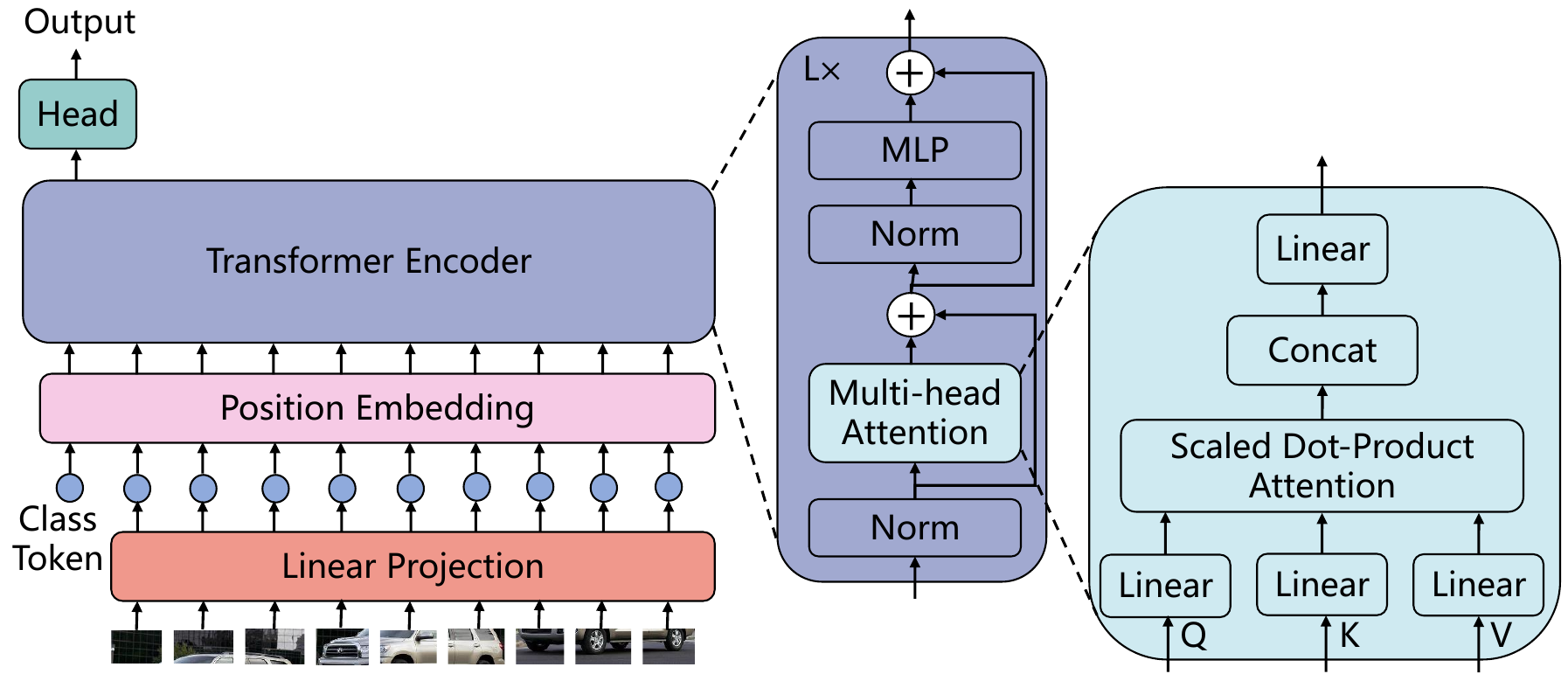}
      \caption{The vision transformer model. The image is split into patches, projected with linear network and added with position embeddings. The tokens are forwarded to the transformer encoder that contains normalization, multi-head attention and multi-layer perception.}
      \label{fig:transformer}
    \end{figure}

    We illustrate our idea based on the vision transformer model as shown in Fig.~\ref{fig:transformer}~\cite{dosovitskiy2020image}. In a vision transformer, an image is split into fixed-size patches and transformed into embeddings through linear projection. The image embeddings are added with the position embedding, and the resulting sequence of tokens is fed into the transformer encoder. A transformer encoder is stacked by a number of attention blocks, which include normalization, multi-head attention (MHA) and multi-layer perception (MLP). 

    Multi-head attention allows the model to extract features from different representation spaces, and each space is called an attention head $i$. 
    In attention, there are a query $Q$, a key $K$, and a value $V$, and $Q, K, V\in \mathbb{R}^{n \times d_{\text{model}}}$, where $n$ is the sequence length of tokens and $d_{\text{model}}$ is the feature dimension. 
    The $Q$, $K$ and $V$ are first mapped to a low-dimension space with the projection parameters (i.e., $\text{head}_i = \text{Attn}(QW_i^{Q},KW_i^{K},VW_i^{V})$). Then, we can use the attention mechanism to model the interactions between tokens and extract the semantic features. We first calculate the attention weight between the query $Q_i$ and the key $K_i$, and apply it to $V_i$ to get a new representation (i.e., $\text{Attn}(Q_i,K_i,V_i) = \text{softmax}(\frac{Q_iK_i^{T}}{\sqrt{d_k}})V_i$).
    $Q_i, K_i \in \mathbb{R}^{n \times d_k}$, $V_i \in \mathbb{R}^{n \times d_v}$, and $d_k, d_v$ are the feature dimensions. In self-attention, $Q$, $K$ and $V$ are equal to the input $x$ at each layer.
    The concatenated attention head is further processed by a linear module, and the result is forwarded to the subsequent block.

    Instead of training from scratch, researchers tend to apply large-scale pre-trained transformer models to various tasks, which has been widely demonstrated to be a practical paradigm for AI development~\cite{bommasani2021opportunities}. 
    However, a large model size leads to high inference latency and makes it challenging to serve queries with burst query loads.

    A common-used method for elastic serving is \emph{model adaptation} that pre-trains multiple versions of models and dynamically loads an appropriate one during runtime. However, it is infeasible for transformers due to the exorbitant training costs and large I/O delay. In this paper, we explore a novel design for elastic transformer serving by utilizing the inherent characteristic of attention: its ability to support token sequences of varying token lengths. Specifically, we propose \emph{token adaptation} that improves accuracy by token prompting and accelerates inference by token reduction. We illustrate how to enlarge and compress the attention space for elastic serving in section~\ref{secprompt} and section \ref{secreduce}.

\subsection{Improve Accuracy by Token Prompting}
\label{secprompt}

     \begin{figure}[t]
    \centering
	\subfloat[Prompt learning for transformer models.]
        {\includegraphics[width = 0.16\textwidth]{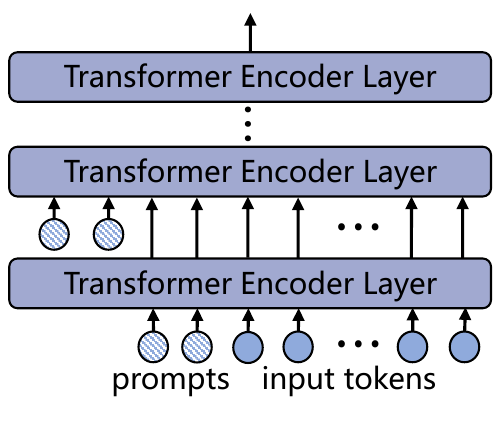}\label{fig:prompt}}
	\hfill
	\subfloat[The token merging mechanism~\cite{bolya2022tome}.]
        {\includegraphics[width = 0.3\textwidth]{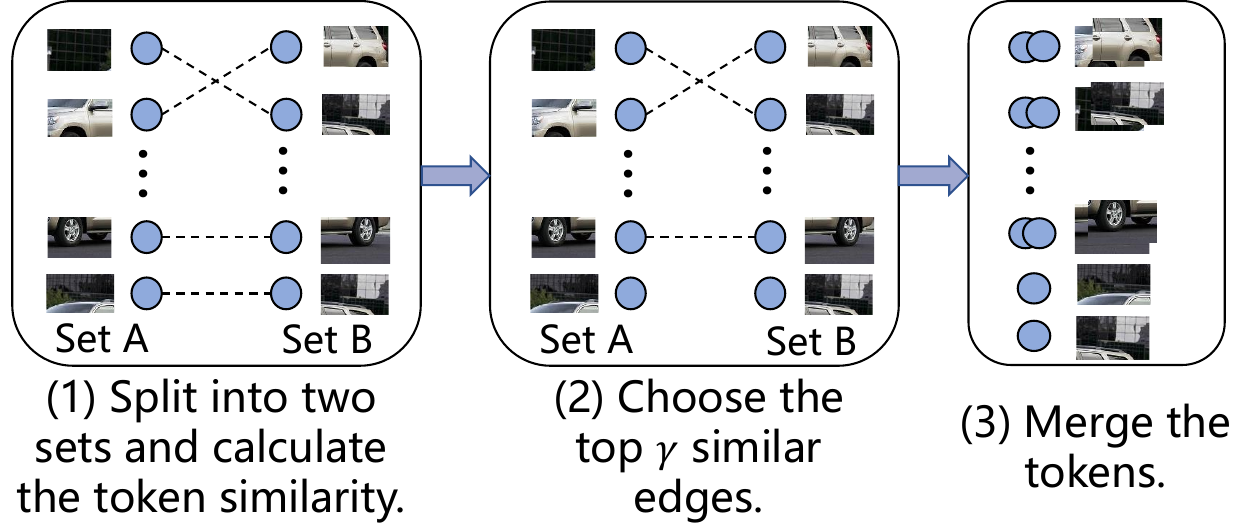}\label{fig:tokenmerge}}
    \caption{Token prompting and token merge.}
    \label{fig:prompt_merge}
    \end{figure}
    

    One effective method for enriching the attention space is to add pre-trained tokens through prompt learning.
    Prompt learning is a parameter-efficient fine-tuning approach for pre-trained transformer~\cite{ding2021openprompt}. 
    As shown in Fig.~\ref{fig:prompt}, prompts are added before the input tokens, and the concatenated tokens are forwarded to the transformer layer to conduct multi-head attention. 
    These prompts are initialized randomly and trained using stochastic gradient descent. During inference, well-trained prompts can be directly prepended to samples.


    Before using prompts for token adaptation, we need to understand the service-level characteristics of prompt learning, i.e., the accuracy and throughput. 
    We train prompts for a pre-trained ViT model on CIFAR10 and CIFAR100 dataset~\cite{krizhevsky2009learning} and set the prompt number as \{2, 4, 8, 16, 32\}. We sampled 1/5 of the training set as the profiling set and evaluated the performance both on the profiling set and testing set. 
    As depicted in the right part of Fig.~\ref{fig:pre_two_acc}, the accuracy exhibits a sharp increase when two prompts are added per layer, and slightly rises as the number of prompts increases further. 
    Prompt learning has greater benefits for difficult tasks, such as CIFAR100.
    Fig.~\ref{fig:pre_two_throughput} shows the throughput result on an NVIDIA GeForce RTX 4080 machine. 
    When the prompt number is increased (0$\sim$32), there is a declining trend in the serving throughput, decreasing from 580 Req/s to 220 Req/s.

    \begin{figure}[t]
    \centering
	\subfloat[Accuracy Comparison.]
        {\includegraphics[width = 0.23\textwidth]{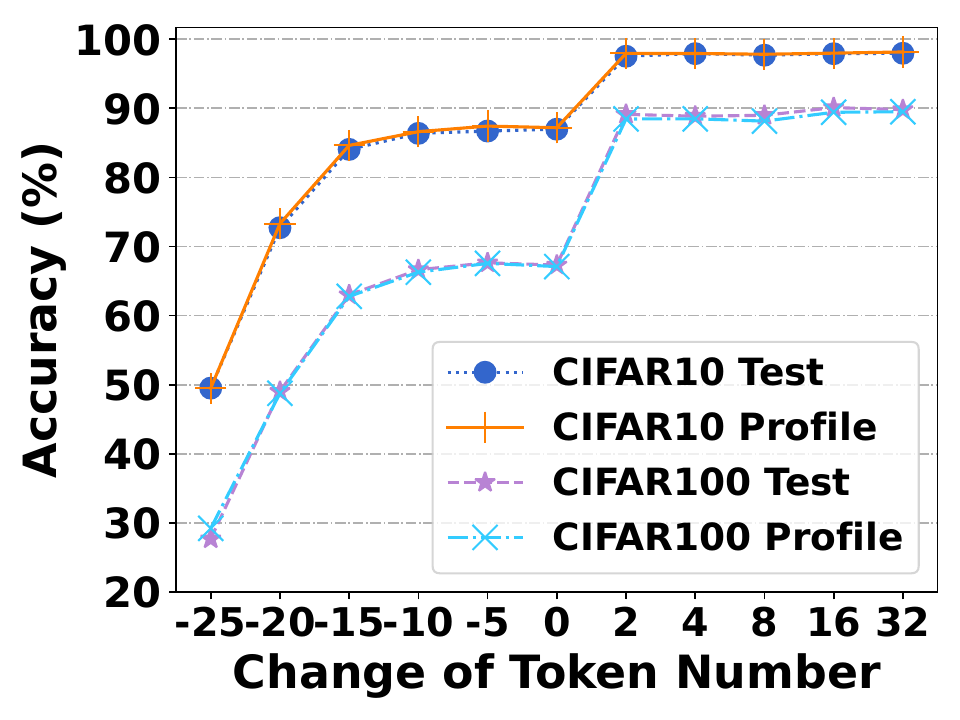}\label{fig:pre_two_acc}}
	\hfill
	\subfloat[Throughput Comparison.]
        {\includegraphics[width = 0.23\textwidth]{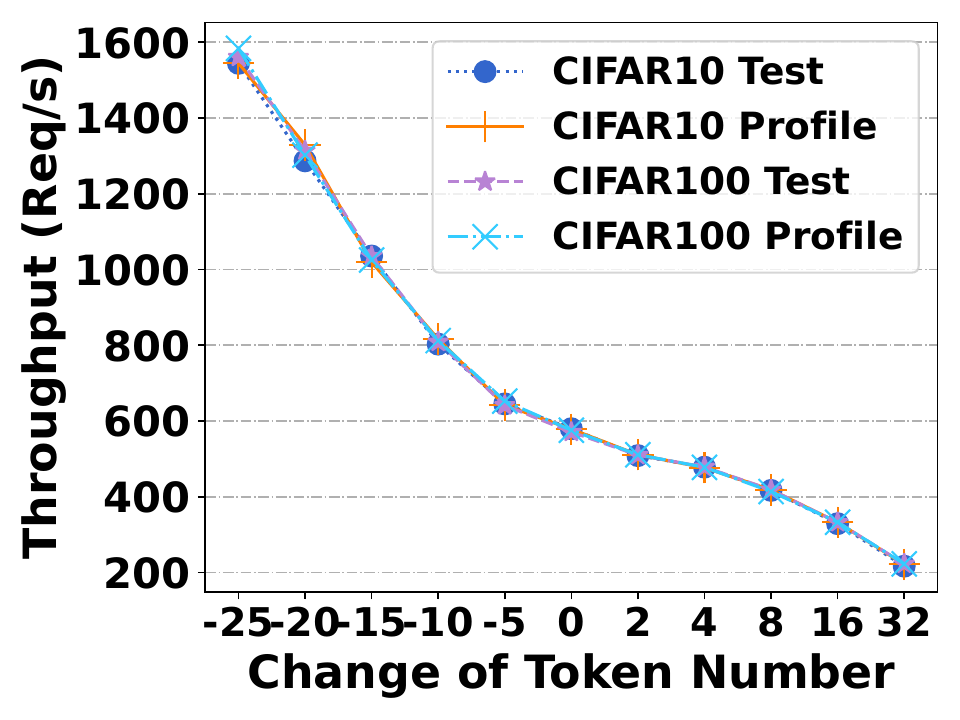}\label{fig:pre_two_throughput}}
    \caption{Accuracy and throughput comparison when we change the numbers of tokens.}
    \label{fig:pre_two_acc_thr}
    \end{figure}
    
    To fully leverage the benefits of prompts, we need to find a sweet spot between accuracy and latency for the incoming requests that align with the request burden, task type, and hardware resources.  

    
\subsection{Accelerate Inference by Token Reduction}
\label{secreduce}


    One effective method for compressing the attention space is to remove redundant or unnecessary tokens, which is useful for accelerating inference. 
    In this paper, we reduce token number by merging similar tokens during inference. 
    We illustrate a state-of-the-art token merging method called ToMe~\cite{bolya2022tome} in Fig.~\ref{fig:tokenmerge}. 
    Firstly, the tokens are split into two sets, and samples in set A pick the most similar sample in set B. Assuming we merge $\gamma$ tokens per layer, the second step is to keep only $\gamma$ edges with the highest similarity values. Finally, similar tokens are merged using a weighted average and concatenated into a new sequence.
    

    We explore the performance of token merging with ToMe on the CIFAR10 and CIFAR100 datasets. We use the vision transformer pre-trained on ImageNet 21K~\cite{ridnik2021imagenet21k} as the backbone and set the merging number from -25 to 0. As shown in  Fig.~\ref{fig:pre_two_acc}, reducing the merging number results in a slight decrease in accuracy, but the trend changes once the merging number drops below -15. At this point, the accuracy declines sharply to 50\% and 28\% for the two datasets.
    The changes of throughput are shown in Fig.~\ref{fig:pre_two_throughput}, which reveals that the throughput has a gradual decline from 1500 Req/s to 580 Req/s (-25$\sim$0). The overhead of the merging algorithm is less than 20ms.

    Though token reduction has its noticeable advantages, it is hard to determine a suitable merging ratio regarding various request inputs, query load and computational resources. 

    
\section{\textsc{OTAS}: Online Token Adaptation System}

    Based on the above analysis, we propose \textsc{OTAS}, a unified framework that achieves online token adaptation for transformer serving. Our framework can autonomously change the token number during the serving period, which allows for efficient and adaptive serving of transformer models.


\subsection{System Design}
\label{sysdesign}
    \begin{figure}[t]
      \centering
      \includegraphics[width=0.9\linewidth]{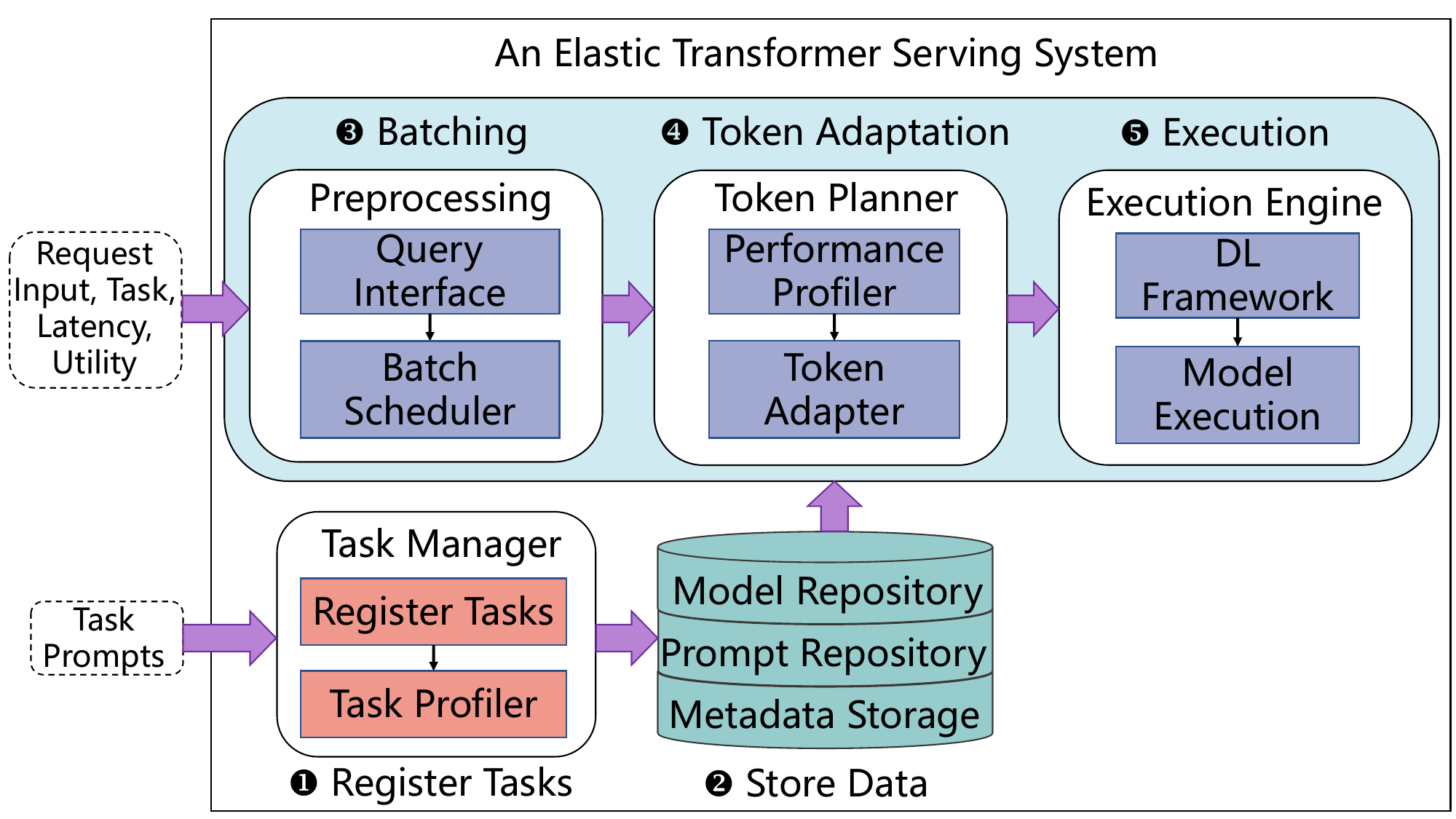}
      \caption{The framework of \textsc{OTAS}, which can assign a query to a batch and allocate the token number automatically.}
      \label{fig:systemdesign}
      \vspace{-10pt}
    \end{figure}

    We first introduce a new serving framework for elastic transformer inference named \textsc{OTAS}.
    The components of \textsc{OTAS} are shown in Fig.~\ref{fig:systemdesign}. The framework consists of two main workflows: task register and query processing. 

    To register a new task, the developer should submit the prompts with the required token numbers, and the prompt parameters are stored in the repository. The task profiler calculates the accuracies and inference latencies for different token numbers and batch sizes on the target device. The profiling data is stored in the metadata storage for future use. 

    The system is designed to handle incoming queries with varying arrival times, inputs, tasks, utilities, and latency requirements. When a query is received, it is added to a batch using a batch scheduler. The batching strategy is described in section~\ref{adapbatch}.
    The resulting batch may contain queries from different tasks with varying utilities and latency requirements. The batch is then stored in a batch queue and awaits execution. The performance profiler is used to predict the accuracy and inference time for different token number settings. The token adapter module uses the profiling data to allocate token numbers for the batches, which is illustrated in section~\ref{onlineadapt}. Finally, the execution engine is responsible for sequentially executing the batches with a transformer model.
 
\subsection{Model Design}
\label{modeldesign}

    To support flexible transformer inference, we propose a unified model that supports token prompting and reduction. As shown in Fig.~\ref{fig:modeldesign}, the prompting module is added before the normalization, and the merging module is added before the MLP. We define the token number change per layer as $\gamma$, where $\gamma>0$ means adding the prompting tokens and $\gamma<0$ means removing some useless tokens.

    The prompting tokens are trained offline and stored in the prompt repository. 
    A token pair is associated with a task and a prompt number, which serves as its index.
    We first initialize the prompt repository randomly and train each token pair separately. 
    We acquire a token pair from the prompt repository at every training epoch and concatenate them with the input tokens. If the batch size is $n_b$ and the input token length is $n_i$, the token shape becomes $n_b\times (n_i+\gamma)$ after prompting. The concatenated tokens can be forwarded to the next module. 
    During inference, the model uses the well-trained prompt parameters in the repository directly.
    The added prompting tokens can inspire the multi-head attention to generate a better result.
    Regarding token merging, the model directly processes the input tokens. Given the token similarity obtained from multi-head attention and a merging rule, the merging module can reduce the token shape from $n_b \times n_i$ to $n_b \times (n_i-|\gamma|)$.
    
    We insert these two modules at each layer. To simplify the design, we assume the model can perform either token prompting or token reduction during inference. Different tasks also have specific head parameters, and the model forwards the sample to its corresponding head to obtain an appropriate prediction probability.

\begin{figure}[t]
      \centering
      \includegraphics[width=\linewidth]{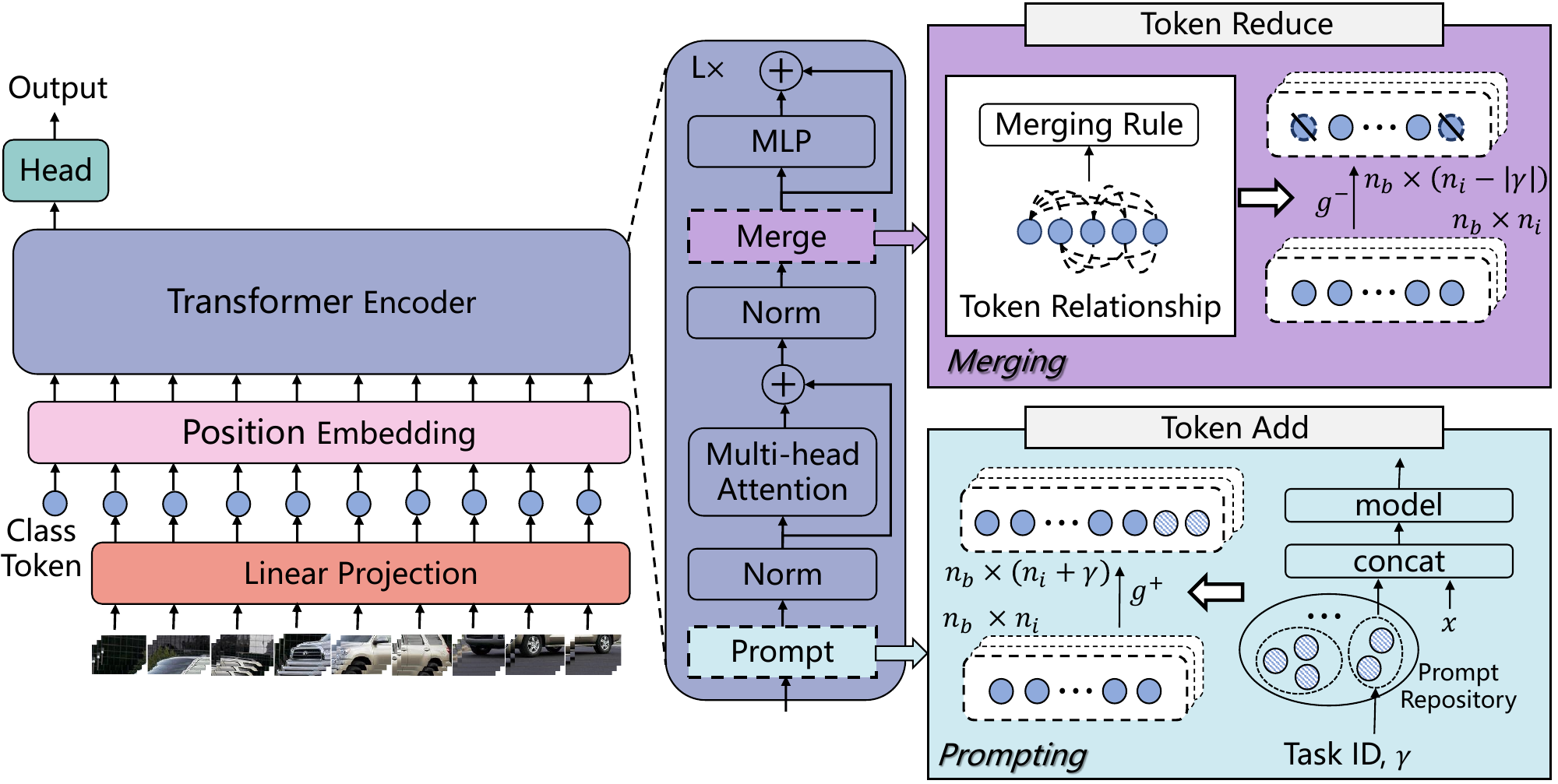}
      \caption{A unified transformer model that incorporates token prompting and token merging.}
      \label{fig:modeldesign}
    \end{figure}
    
\subsection{Adaptive Batching}
\label{adapbatch}

    Batching queries for inference can improve the system's throughput by making full use of the computational capabilities of the device and reducing the costs associated with model initialization and data communication~\cite{crankshaw2017clipper,cui2022dvabatch}.  Fig.~\ref{fig:pre_batch} illustrates the advantages of batching, where we evaluate the throughput with batch sizes ranging from 1 to 64. 
    The throughput shows a rapid increase as the batch size is increased. For example, when $\gamma=-15$, the throughput increases from 100 Req/s to 1000 Req/s and converges when the batch size is 20.  
    Therefore, batching can significantly improve throughput and enable more efficient query processing.
    
    \begin{figure}[t]
    \centering
	\subfloat[Throughput comparison on CIFAR10.]
        {\includegraphics[width = 0.23\textwidth]{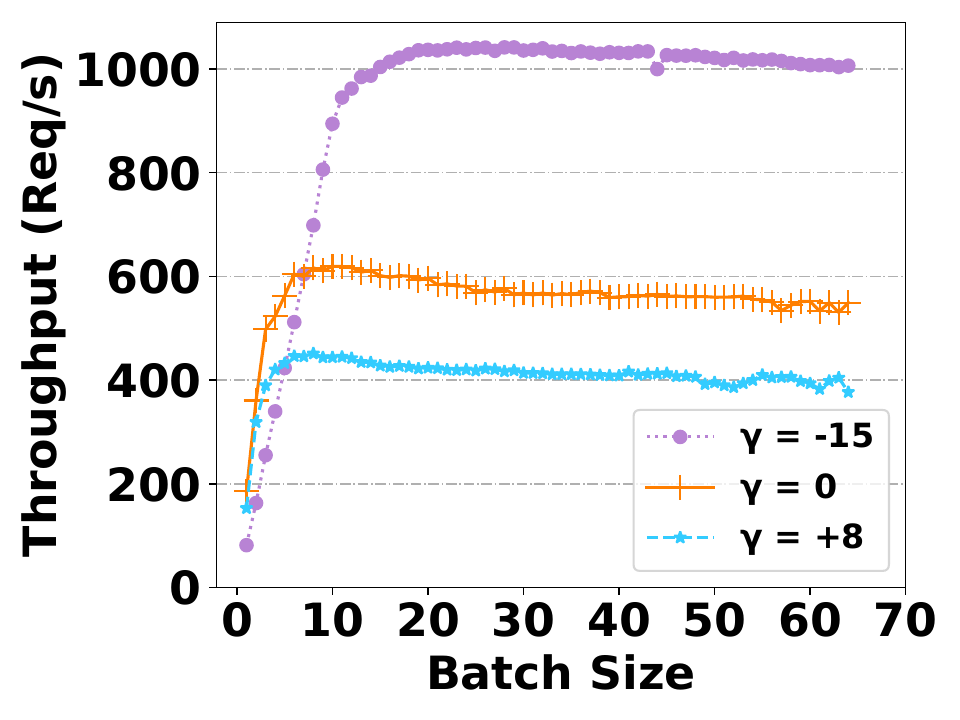}\label{fig:batch_10}}
	\hfill
	\subfloat[Throughput comparison on CIFAR100.]
        {\includegraphics[width = 0.23\textwidth]{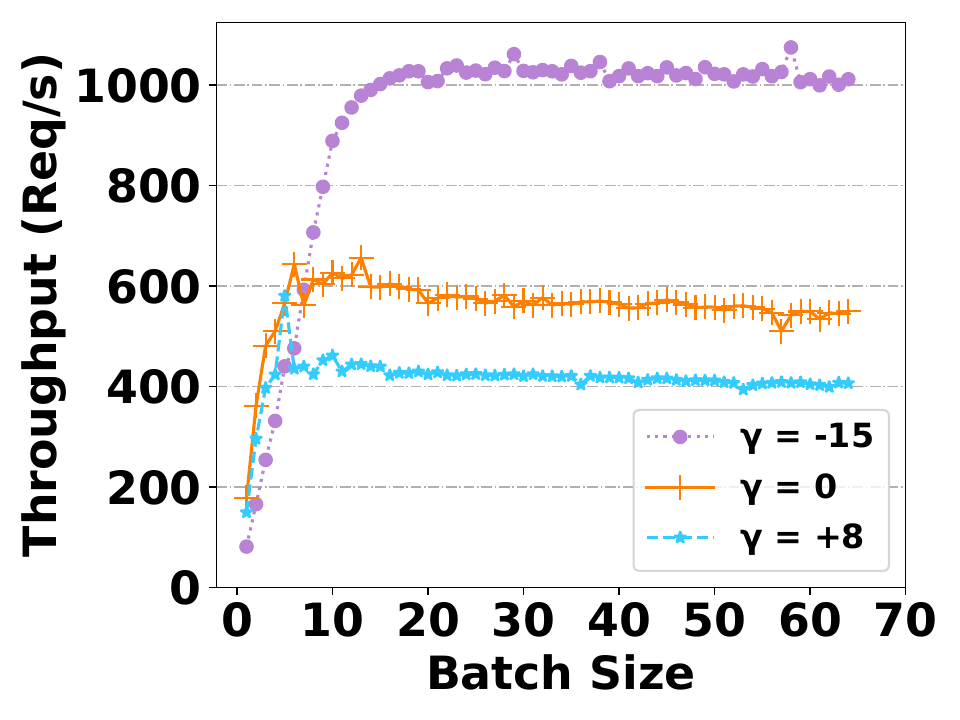}\label{fig:batch_100}}
    \caption{Throughput comparison of different batch sizes.}
    \label{fig:pre_batch}
    \end{figure}

    \begin{algorithm}[t]
    \caption{Batching Algorithm: Adding a Query into a Batch.} 
    \label{alg_batch} 
    \KwIn{Batch Queue $B$, Query $r$;}
    \KwOut{Batch Queue $B$;}
    \For{$b \in [N_B,1]$}{
        \If(\tcp*[f]{Arrival time;}){$s_b + \delta < s_r$}{ 
        break\;
        }
        \ElseIf(\tcp*[f]{Batch size;}){$|B_b| >= \epsilon$}{
            continue\;
        }
        \ElseIf(\tcp*[f]{Finish time;}){$\Vert d_b - d_r\Vert>\eta$}{
            continue\;
        }
        \ElseIf(\tcp*[f]{Utility;}){$\Vert u_b - u_r\Vert>\mu$}{
            continue\;
        }
        $B_b$.add($r$); \tcp*[f]{Find a batch $b$;} \\
        \Return $B$\;
    }
    $B$.add(\{$r$\});  \tcp*[f]{Create a new batch;} \\
    \Return $B$\;
    \end{algorithm}
    
    While batching has clear benefits, one challenge is to design a batching strategy that effectively groups similar requests together. To achieve this, we batch incoming queries based on their similar arrival patterns and service-level objectives, such as latency constraint and utility.
    We use the notation $r$ to represent a request, where $s_r$, $l_r$, $d_r$, and $u_r$ represent the request's arrival time, latency requirement, finish deadline, and utility, respectively, such that $d_r=s_r+l_r$.

    The grouped queries are stored in the batch queue $B$, and we denote the $b$-th batch as $B_b$. The arrival time of a batch $b$ is defined as the earliest arrival time among its requests, i.e., $s_b=\text{min}\{s_r\}, r\in B_b$. Similarly, the finish deadline of a batch $b$ is defined as the earliest required finish time among its requests, i.e., $d_b=\text{min}\{d_r\}, r\in B_b$. 

    The batching algorithm is described in Algorithm~\ref{alg_batch}, which assigns a query to the current batches or initializes a new batch. The key idea of the algorithm is constructing a batch with constraints on batch size, arrival time, utility and deadline. 
    Specifically, the algorithm ensures that the waiting time of the first request in a batch is less than $\delta$, the batch size is smaller than a pre-defined threshold $\epsilon$, and the deadline difference between the batch and the query $r$ is not larger than a threshold $\eta$.  We use $u_b$ to represent the utility of the first arrival query in a batch $b$, and restrict the utility value for subsequent incoming query $r$ to be close to the value of $u_b$ with a threshold $\mu$.
    These constraints ensure that queries with similar arrival patterns and service-level objectives can be processed together, which is beneficial for token adaptation.
    If a batch that meets the constraints for the incoming query is found, the query is added to that batch $b$ (Line 1$\sim$9). Otherwise, a new batch is created for the query and added to the batch queue.

\subsection{Online Token Adaptation}
\label{onlineadapt}

    After constructing the batch queue, the next step is to assign token adaptation schemes for batches. In this section, we present an optimization problem for token adaptation and propose a dynamic programming algorithm to obtain the solution.
    
\subsubsection{Problem Formulation}

     We define the token adaptation scheme for a batch $b$ as $\gamma_b$, where $\gamma_b<0$ indicates reducing the token number, $\gamma_b>0$ indicates adding some prompting tokens, and $\gamma_b=0$ indicates making the inference with the vanilla transformer model. 
     $\gamma$ is a discrete value that can be selected from a pre-defined list.
     If the serving system successfully provides an accurate result for request $r$ under the latency requirement, the system can be rewarded with utility $u_r$, such as the money. We use $\alpha_r \in \{0,1\}$ to represent whether it successfully serves query $r$. The required memory of batch $b$ and the available GPU memory are denoted as $M_b$ and $M_{\text{GPU}}$.

    The optimization problem is defined in Eq.~\eqref{equa:optim}, where the goal is to allocate the token change number $\gamma$ to maximize the overall utility for all requests. Constraint~\eqref{equa:optim-a} ensures that all requests can be completed within their respective deadlines, where $t_r^{(q)}$ and $t_r^{(p)}$ are the queuing time and processing time. 
    Constraint~\eqref{equa:optim-b} ensures that the batches are executed sequentially. Constraint~\eqref{equa:optim-c} imposes a memory restriction, as larger batch sizes and prompt numbers may increase the memory demand.
    \begin{align}
    \mathop{\max}\limits_{\gamma_b} &\sum_{b \in [1,N_B]} \sum_{r \in B_b} u_r \cdot \alpha_r; \label{equa:optim}\\
    s.t. \quad & s_r + t_r^{(q)} + t_r^{(p)} < d_r, \forall r \in B_b; \tag{\ref{equa:optim}{a}} \label{equa:optim-a}\\
    & s_r + t_r^{(q)} + t_r^{(p)} < s_{r'} + t_{r'}^{(q)}, \forall r \in B_b \text{ and } \forall r' \in B_{b+1}; \tag{\ref{equa:optim}{b}} \label{equa:optim-b}\\
    & M_b < M_{\text{GPU}} \tag{\ref{equa:optim}{c}} \label{equa:optim-c}.
    \end{align}

    The above problem formulation considers both the query load and request characteristics. If the batch queue has a high volume of queries, we should pick a smaller $\gamma$ to reduce the queuing and processing time and serve more requests.
    Conversely, we can increase the value of $\gamma$ to derive an accurate result and earn more utilities.
    Then, we analyze the NP-hard property of problem~\eqref{equa:optim}.
    
    \begin{theorem}
    \label{nphard}
    The problem~\eqref{equa:optim} is an NP-hard problem.
    \end{theorem}

    \begin{proof}
    The token adaptation problem is an NP-hard problem because it can be reduced from another
    NP-hard problem--Weighted Interval Scheduling Problem (WISP)~\cite{kolen2007interval}. Given a set of intervals
    with a weight, the objective of WISP is to select some intervals that can maximize the sum of
    the weights while the selected intervals are pairwise disjoint. We can transform
    our problem to the WISP. We consider each batch as an interval with a weight equal to its utility.
    Our goal is to efficiently process the batches so that the sum of utilities is maximized.  
    Our problem is more difficult than WISP because we also need to adjust the running time for the picked intervals with different $\gamma$ values.
    \end{proof}
    
\subsubsection{Algorithm Design}

    Due to the NP-hardness of the above problem, we propose an efficient dynamic programming algorithm to derive the solution in Algorithm~\ref{alg_allocate}, which takes the batch queue $B$, current time $T$, the available $\gamma$ list and the estimated arriving rate $q$ as inputs and outputs the updated batch queue with allocated token number $\gamma$. The key idea is to find the largest utility value for a batch $b$ with a $\gamma_b$ through iterative traversal.

    \begin{algorithm}[t]
    \caption{Autonomous Token Adaptation Algorithm} 
    \label{alg_allocate} 
    \KwIn{Batch Queue $B$, Clock Time $T$, Selection List $L^{(\gamma)}$, Incoming Request Rate $q$;}
    \KwOut{Batch Queue $B$;}
    Sort($B$) according to $d_b$\;
    \If{$N_B \leq \beta$ or \text{initial\_stage}=\text{True}}{
        $B$ = Manually\_Allocate($B,T,L^{(\gamma)},q$) \;
        \Return $B$\;
    } 
    Initialize $dp \in \mathbb{R}^{(N_B+1)\times (N_{\gamma}+1)}$ by $0$\;
    Initialize $S \in \mathbb{R}^{(N_B+1)\times (N_{\gamma}+1)}$ by $1$\;
    Initialize $C \in \mathbb{R}^{(N_B+1)\times (N_{\gamma}+1)}$ by $T$\;
    Initialize $J \in \mathbb{R}^{(N_B+1)\times (N_{\gamma}+1)}$ by $0$\;
    \For{$b \in [1,N_B]$}{
        \For{$l_{b} \in [0,N_{\gamma}]$}{ 
            \For{$l_{b-1} \in [0,N_{\gamma}]$}{
                \If{$dp[b-1,l_{b-1}] == -\infty$}{
                    continue;
                } 
                \If{$l_{b}==0$}{
                    \If{$dp[b-1,l_{b-1}]>dp[b,l_{b}]$}{
                        $dp[b,l_{b}]=dp[b-1,l_{b-1}]$\; 
                        $S[b,l_{b}]= l_{b-1}$\;
                        $C[b,l_{b}]= C[b-1,l_{b-1}]$\; 
                        $J[b,l_{b}]=1$\;
                    }
                }
                \Else{
                    $\gamma_b = L^{(\gamma)}[l_{b}]$\;
                    $\hat{t}_r^{(p)}, \hat{U}_b = \text{Profile}(B_b, \gamma_b)$\;
                    \If{$C[b-1,l_{b-1}]  + \hat{t}_b^{(p)} < d_b$}{
                        $u=dp[b-1,l_{b-1}] + \hat{U}_b$\;
                        $J[b,l_b]=1$\;
                        \If{$u > dp[b,l_b]$}{
                            $dp[b,l_b]=u$\;
                            $S[b,l_b]= l_{b-1}$\;
                            $C[b,l_b]=C[b-1,l_{b-1}] + \hat{t}_b^{(p)}$\;
                        }
                    }
                }
            }  
            \If{$l_b>0$ and $J[b,l_b]==0$}{
                $dp[b,l_b] = -\infty$; \\
                $C[b,l_b] = +\infty $\;
            }
        }
    }
    $l=\arg\max dp[N_B]$; \\  
    $B_{N_B}.\gamma = L^{(\gamma)}[l]$\;
    \For{$b=N_B-1$ to 1}{
        $l = S[b+1,l]$\;
        $B_{b}.\gamma = L^{(\gamma)}[l]$\;
    }
    \Return $B$\;
    \end{algorithm}
    
    \begin{algorithm}[t]
    \caption{Manually\_Allocate: Allocate $\gamma$ According To The Arriving Rate.} 
    \label{alg_allocate_manual} 
    \KwIn{Batch Queue $B$, Clock  Time $T$, Selection List $L^{(\gamma)}$, Incoming Request Rate $q$;}
    \KwOut{Batch Queue $B$;}
    $\gamma$ = $f(q)$\;
    \For{$b \in [1,N_B]$}{
        $\hat{t}_r^{(p)} = \text{Profile}(B_b, \gamma)$;  \\
        \If{$T+ \hat{t}_r^{(p)} \geq d_b$}{
            $B_b.\gamma \gets \min(L^{(\gamma)}$);
        } 
        \ElseIf{$\overline{U_b} > \kappa$}{
            $B_b.\gamma \gets \max(L^{(\gamma)}$);
        }  
        \Else{
            $B_b.\gamma \gets \gamma$;
        }
        $\hat{t}_r^{(p)} = \text{Profile}(B_b, B_b.\gamma)$\; 
        $T = T+\hat{t}_r^{(p)}$;
    }
    \Return $B$\;
    \end{algorithm}

    We begin by sorting the batches according to their required deadlines. If the size of the batch queue is less than a threshold $\beta$ or the serving system is in the initial stage, we allocate the token number based on the query load with Algorithm~\ref{alg_allocate_manual}. This is because the dynamic programming algorithm works well when there are sufficient batches to make a long-term schedule. Algorithm~\ref{alg_allocate_manual} allocates the token number $\gamma$ by comparing the incoming request rate $q$ and the throughput of different $\gamma$ values.  We calculate the arriving rate $q$ within the previous inference window and apply a function $f$ to map $q$ to a suitable value of $\gamma$ (Line 1). $f$ can be profiled offline according to the throughput of different $\gamma$ values. Then, we adjust the selection of $\gamma$ according to the query characteristics. 
    We predict the execution time for a batch $b$. If the estimated completion time exceeds the deadline, we set the token number as the minimum value to meet the latency constraints (Line 3$\sim$5). 
    If the average utility $\overline{U_b}$ is larger than a threshold $\kappa$, we set the token number as the maximum value to prioritize the critical queries. Finally, we estimate the execution time and update current time $T$.  

    Algorithm~\ref{alg_allocate} utilizes four auxiliary arrays of size $(N_B+1)\times(N_{\gamma}+1)$ to implement dynamic programming, where $N_B$ and $N_{\gamma}$ are the sizes of batch queue and the number of available $\gamma$ values. Specifically, $dp$ records the accumulated utilities, $S$ records the previous $\gamma$ selection scheme, and $C$ records the clock time after executing batch $b$ with $\gamma$. The array $J$ indicates whether executing $b$ with $\gamma$ satisfies the deadline requirement. 
    For each batch in the batch queue, we iteratively assign a value of $\gamma$ from the list $L^{(\gamma)}$ to batch $b$ using the index $l_b$ (Line 9$\sim$11). 
    If batch $b-1$ cannot be executed with $\gamma$ indexed with $l_{b-1}$, we continue to the next iteration of the loop (Line 12$\sim$13). When the value of $l_b$ is 0, it indicates that batch $b$ is not executed, and we directly find a larger utility value from batch $b-1$ and assign it to batch $b$ (Line 14$\sim$19).

    When executing $\gamma_b$ for batch $b$, we first estimate inference time and utility through profiling (Line 22). If the completion time is smaller than the required deadline, we calculate the overall utility and set the execution plan as 1 (Line 23$\sim$25). If the utility is larger than the previous values, we update the matrixes. If there is no feasible execution plan for batch $b$ with $\gamma_b$, we set the $dp$ value as $-\infty$ and the clock time as $+\infty$ (Line 30$\sim$32).
    
    Once we have calculated the utility values and their corresponding choices, we can derive the solution by backtracking. We first determine the value of $\gamma$ for the $N_B$-th batch based on the highest $dp$ value. For each batch, we obtain the index of $\gamma$ according to the value of $S[b+1,\gamma]$. Finally, we return the updated batch queue $B$.

    We estimate the execution time and utility for a batch $b$ with the profiling data. We profile the accuracy and sample-level inference latency for all tasks and store them in the metadata storage.
    To estimate the inference time for the current batch, we first count the number of samples for each task, and then multiply the sample number by the corresponding profiling inference time to obtain the execution time for that task. We then sum up the calculation results for all tasks to obtain the predicted inference time of a batch.
    To calculate the overall utility, we compute the product of the accuracy with a selected $\gamma$ and the utility of each query in the batch. Then, we sum up the product result of all queries to obtain the total utility of a batch. 
    During profiling, we ensure that all the running processes adhere to the memory constraints of Eq.~\eqref{equa:optim-c}.
    



\section{Implementation}


\textbf{\textsc{OTAS} Description.}
We provide four data structures and corresponding interfaces to implement the \textsc{OTAS}. 
\textsf{TransformerModel} is a transformer model class that comprises token prompting and token reduction modules. This model is loaded with pre-trained weights. 
\textsf{TaskModel} stores all parameters for a task, such as the prompts and classification head.
\textsf{ServeModel} serves as the base model for the front-end surface. Its \textsf{forward} method accepts a batch of inputs, the corresponding input tasks, the parameter list of \textsf{TaskModel} and the $\gamma$ value as input and returns the inference result.
The \textsf{Batch} class is responsible for adding a query to the batch, providing profiling results and returning a batch of queries within latency constraints for inference.

\textbf{Implementation Tools.}
We implement the \textsc{OTAS} based on the PetS~\cite{280684}.
We use Python to process the incoming queries and implement the batching and token adaptation algorithms.
We use PyTorch to define the neural networks, including \textsf{TransformerModel}, \textsf{TaskModel} and \textsf{ServeModel}.
We build the transformer model with timm library~\cite{rw2019timm} and insert two modules to add and remove the processing tokens at each layer.
We implement the prompt learning and token reduction methods according to VPT~\cite{jia2022visual} and ToMe~\cite{bolya2022tome}.

\textbf{User Interface.}
The system enables users to make a query and register tasks with two interfaces.
The \textsf{Make\_Query} interface processes a query that comprises an image sample and various attributes, such as the task ID, latency requirement and utility. 
Then, the query can be assigned to a batch with Algorithm~\ref{alg_batch}. 
The \textsf{Register\_Task} interface saves the task parameters in the task model list and the corresponding latency and utility values in the task data list.

\section{Experiment}

\textbf{Setup.}
We use the ViT-Base model pre-trained on ImageNet 21K as the foundation model, which contains 12 transformer layers. The head number of attention is 12, and the feature dimension is 768. The patch size of the images is $16\times16$. We use three datasets, including CIFAR10, CIFAR100~\cite{krizhevsky2009learning} and EuroSAT~\cite{helber2019eurosat}, and 1/5 of the training data was randomly selected as the profiling set.
We define the $\gamma$ selection list as \{-20, -15, -10, -5, 0, 2, 4, 8\} and adjust it according to the query rate. The values of $\delta$, $\epsilon$, $\eta$ and $\mu$ in Algorithm~\ref{alg_batch} are set as 0.5s, 64, 0.5s and 0.8 respectively. We set the value of $\beta$ as 5 and define the initial\_stage as the first 2 seconds of the service. The value of $\kappa$ in Algorithm~\ref{alg_allocate_manual} is 0.8.
According to Fig.~\ref{fig:pre_two_throughput}, the $f$ function is defined in Table~\ref{tabel_functionf}.

    \begin{table}[t]
    \centering
    \begin{tabular}{ccccc}
       \toprule
       $\gamma$ & 8 & 4 & 2 & 0  \\
       \midrule
       $q$ (Req/s) & 1$\sim$279 & 280$\sim$319 & 320$\sim$348 & 350$\sim$379  \\
       \bottomrule
    \end{tabular}
    \begin{tabular}{ccccc}
       \toprule
       $\gamma$ & -5 & -10 & -15 & -20 \\
       \midrule
       $q$ (Req/s) & 380$\sim$449 & 450$\sim$519 & 520$\sim$999 & $\textgreater$1000 \\
       \bottomrule
    \end{tabular}
    \caption{The projection function from arriving rate to $\gamma$.}
    \label{tabel_functionf}
    \end{table}

We train the prompts for tasks offline. The training batch size, epochs and learning rate are set as 32, 50 and 0.002.

We evaluate \textsc{OTAS} on an NVIDIA GeForce RTX 4080 (12th Gen Intel(R) Core(TM) i9-12900K CPU) machine. 

\textbf{Baseline.}
We compare \textsc{OTAS} with PetS~\cite{280684} and INFaaS~\cite{romero2021infaas}. 
PetS is a unified framework for serving transformers with parameter-efficient methods and optimizes task-specific and task-shared operators. We remain token unchanged for PetS and perform inference with a shared foundation model and task-specific heads. 
INFaaS is a model adaptation method that selects an appropriate model according to the query load. We set the candidate model list as ViT-Small, ViT-Base and ViT-Large.
We also compare \textsc{OTAS} with ToMe~\cite{bolya2022tome} and VPT~\cite{jia2022visual} that uses fixed merging or prompting number.

    \begin{table}[t]
    \centering
    \begin{tabular}{cccc}
       \toprule
       Query Type & Task & Latency & Utility  \\
       \midrule
       1 & CIFAR10 & 0.6s & 0.3  \\
       2 & CIFAR10 & 1s & 0.01  \\
       3 & CIFAR100 & 0.6s & 1  \\
       4 & CIFAR100 & 1s & 0.2  \\
       5 & EuroSAT & 0.6s & 0.3  \\
       6 & EuroSAT & 1s & 0.01  \\
       \bottomrule
    \end{tabular}
    \caption{The latency and utility of queries.}
    \label{tabel_querytype}
    \end{table}
    \begin{figure}[t]
    \centering
	\subfloat[The query trace on the synthetic dataset in the first 200 seconds.]
        {\includegraphics[width = 0.23\textwidth]{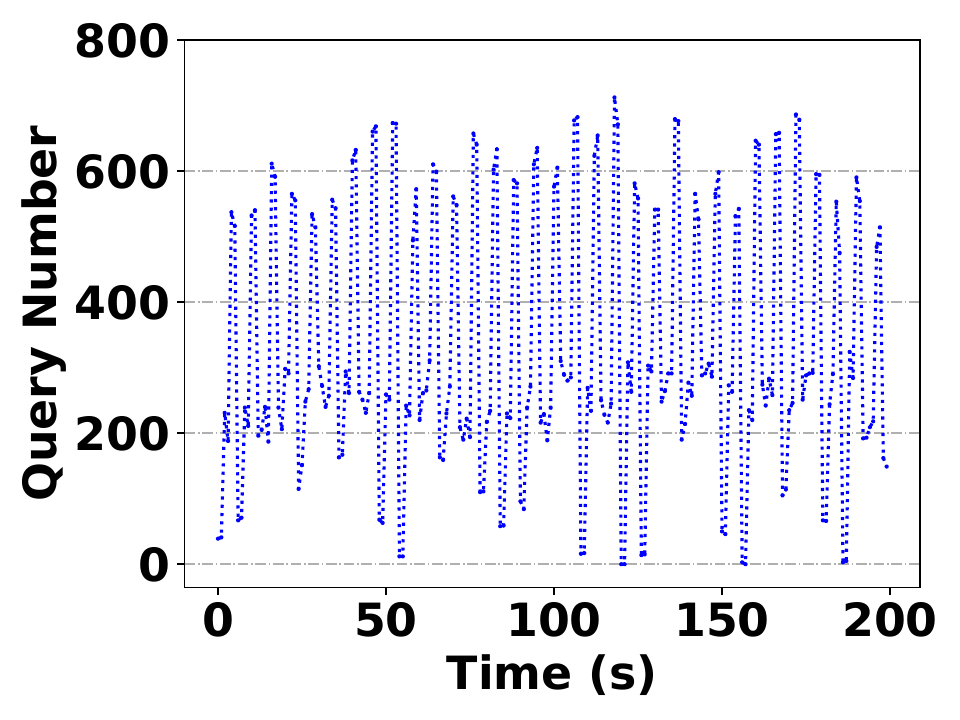}\label{fig:period_query_trace}}
	\hfill
	\subfloat[The query trace on the MAF dataset in the first 1000 seconds.]
        {\includegraphics[width = 0.23\textwidth]{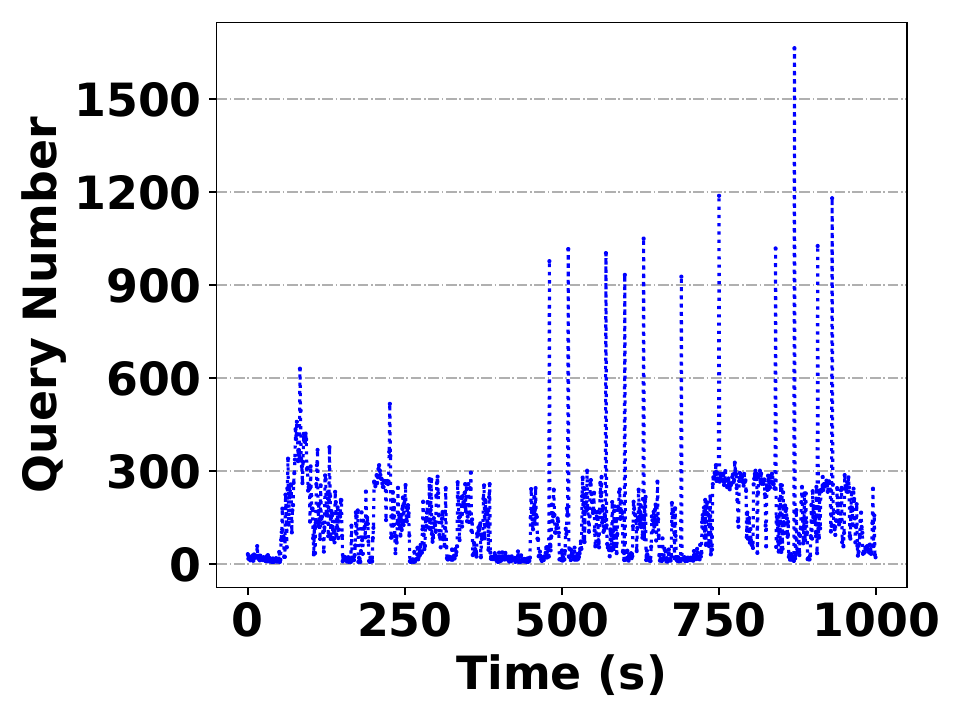}\label{fig:azure_query_trace}}
    \caption{The query trace on two datasets.}
    \label{fig:query_trace}
    \end{figure}

\textbf{Workloads.}
We evaluate the algorithms using both the synthetic query trace and real-world production trace.
For synthetic workloads, we generate the query traces that have fluctuating loads. We randomly generate the arrival time for queries according to the Poisson distribution~\cite{romero2021infaas}. We randomly select a query type from Table~\ref{tabel_querytype} for each query. We conduct experiments over a 30-minute serving period, with more than 63k queries processed. In Fig.~\ref{fig:period_query_trace}, we present the query number per second during the first 200 seconds. The query rate varies between 200 Req/s to 700 Req/s in each second.

For real-world workloads, we use the publicly-released traces of Microsoft collected from Azure Functions in 2021 (MAF)~\cite{10.1145/3477132.3483580,285173,li2023alpaserve}. We select a 120-hour trace for experiments. 
We aggregate requests collected every two-minute interval into one-second interval to create a challenging trace. The query number per second in the first 1000 seconds is presented in Fig.~\ref{fig:azure_query_trace}.
During more than 60\% of the serving period, the query rate remains below 300 Req/s. There are still some instances where the request number per second exceeds 600 Req/s.

\subsection{Main Results}

    \begin{figure}[t]
    \centering
	\subfloat[The utility of different system designs on the synthetic dataset.]
        {\includegraphics[width = 0.23\textwidth]{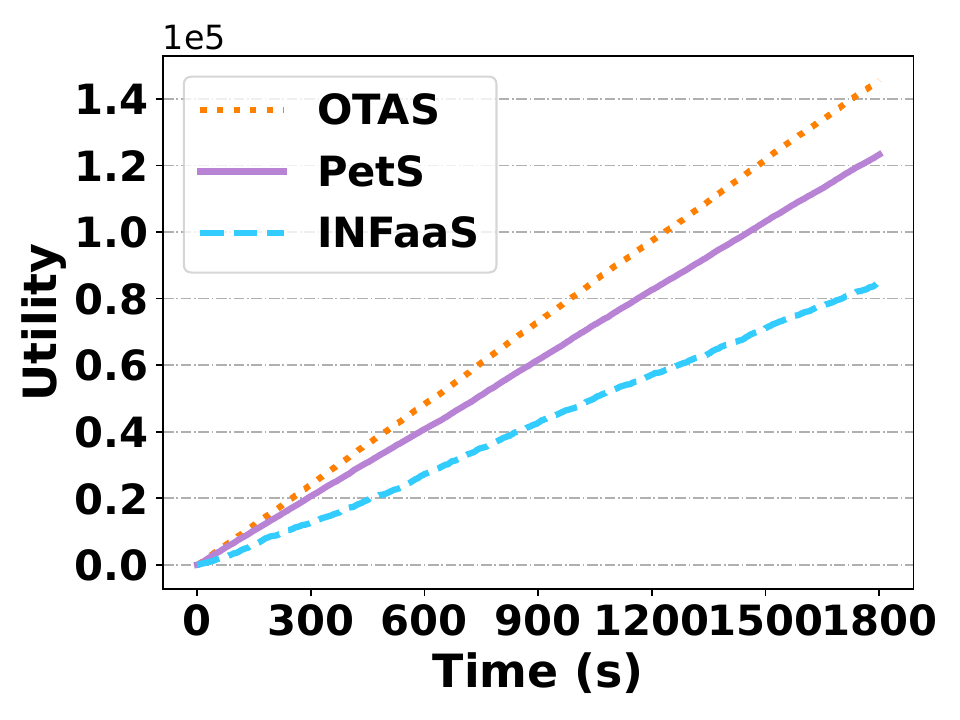}\label{fig:period_cumu}}
	\hfill
	\subfloat[The utility of different system designs on the MAF dataset.]
        {\includegraphics[width = 0.23\textwidth]{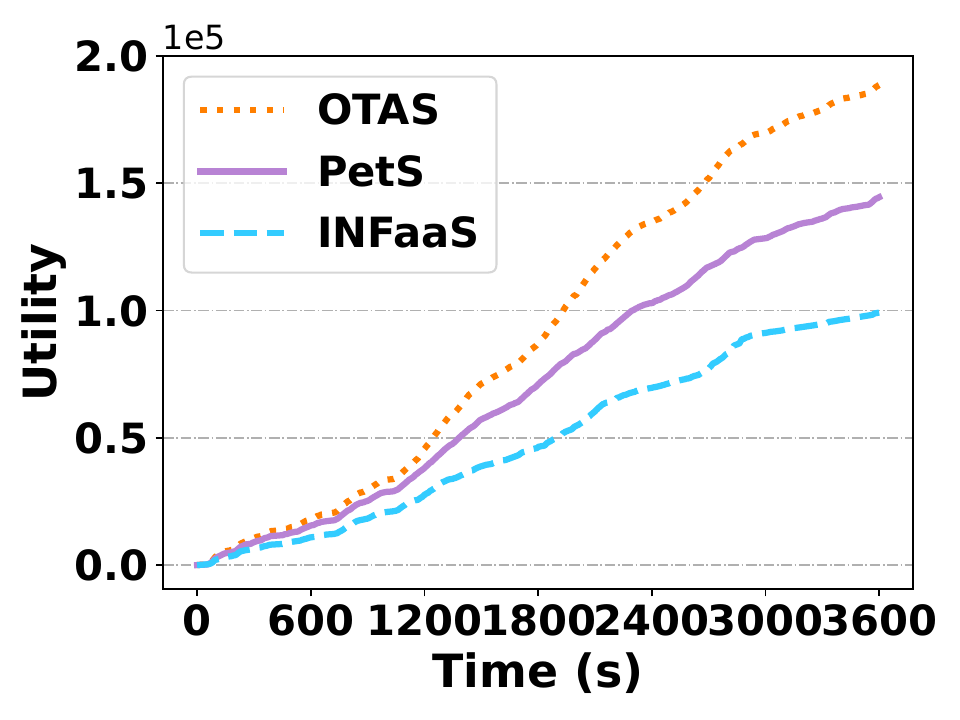}\label{fig:azure_cumu}}
    \caption{The utility comparison of different system designs.}
    \label{fig:utility_cumu}
    \end{figure}
    
    \begin{figure}[t]
    \centering
	\subfloat[The utility comparison of different methods on the synthetic dataset.]
        {\includegraphics[width = 0.23\textwidth]{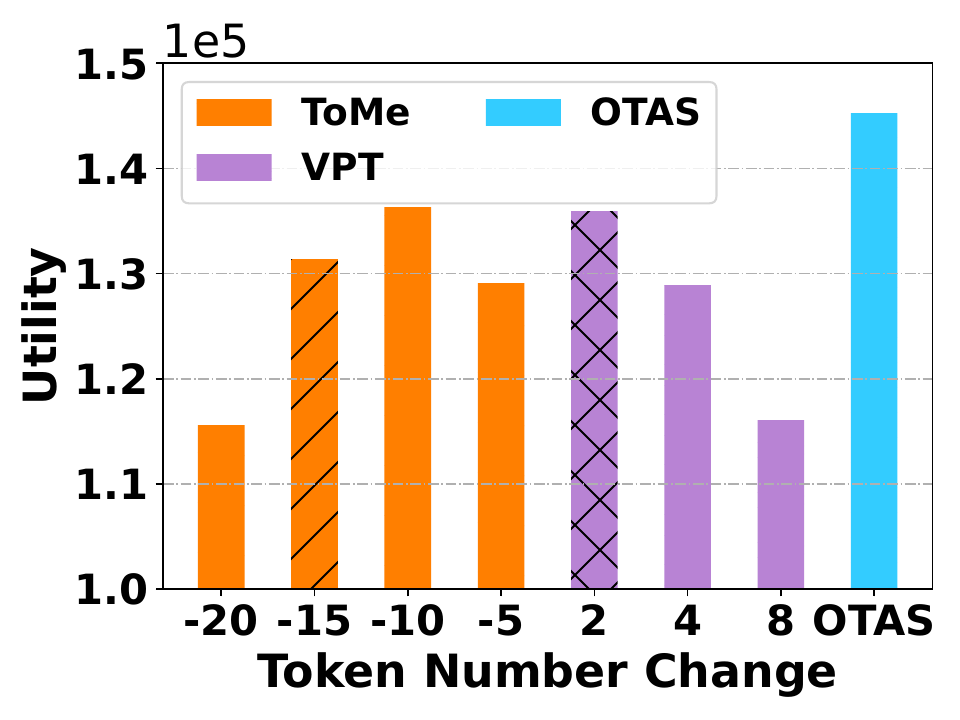}\label{fig:utility_bar}}
	\hfill
	\subfloat[The utility comparison of different methods on the MAF dataset.]
        {\includegraphics[width = 0.23\textwidth]{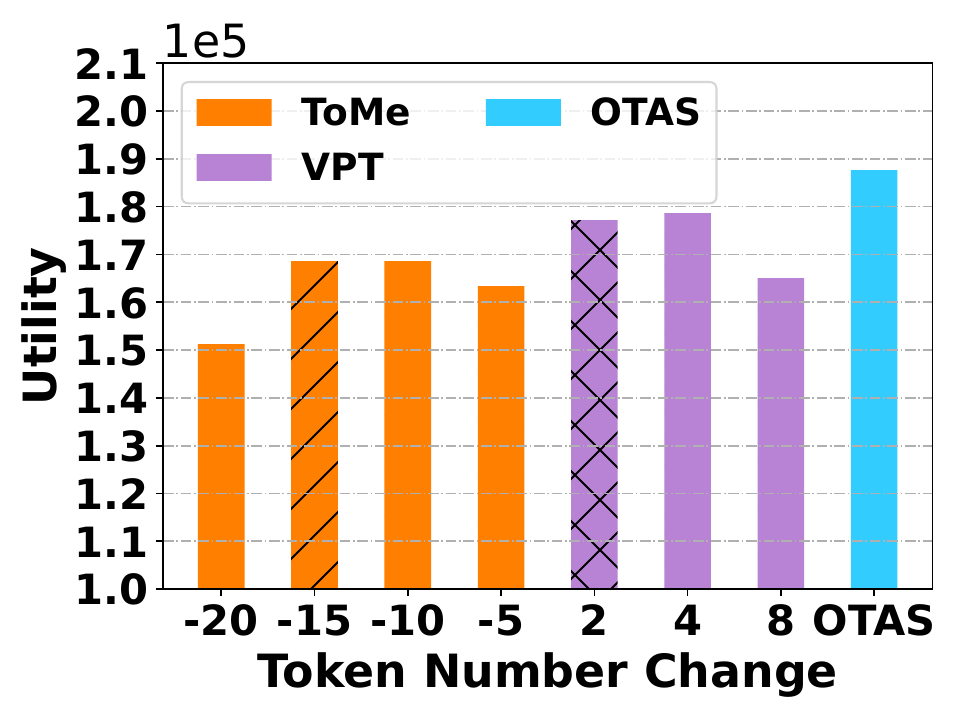}\label{fig:azure_utility_bar}}
    \caption{The utility comparison of different token number.}
    \label{fig:utility_bar}
    \end{figure}
    


    \textbf{The overall utility.} If the system can return an accurate result for a query under the latency constraint, it can be rewarded the utility of the query.
    The accumulated utilities of three system designs on the synthetic dataset are shown in Fig.~\ref{fig:period_cumu}. 
    \textsc{OTAS} obtains about $1.46\times 10^{5}$ utilities and results in a utility improvement of 18.2\% and 72.5\%. INFaaS behaves the worst because it has a long I/O latency to switch the models. The overall utility of the MAF dataset is shown in Fig.~\ref{fig:azure_cumu}. \textsc{OTAS} can improve the utility by up to $90.1\%$. 

    The utility comparison with fixed token number is shown in Fig.~\ref{fig:utility_bar}. \textsc{OTAS} outperforms both ToMe and VPT because it can adjust the token strategy according to the query load.


    \begin{figure}[t]
    \centering
	\subfloat[The CDF plot of accuracies on the synthetic dataset.]
        {\includegraphics[width = 0.23\textwidth]{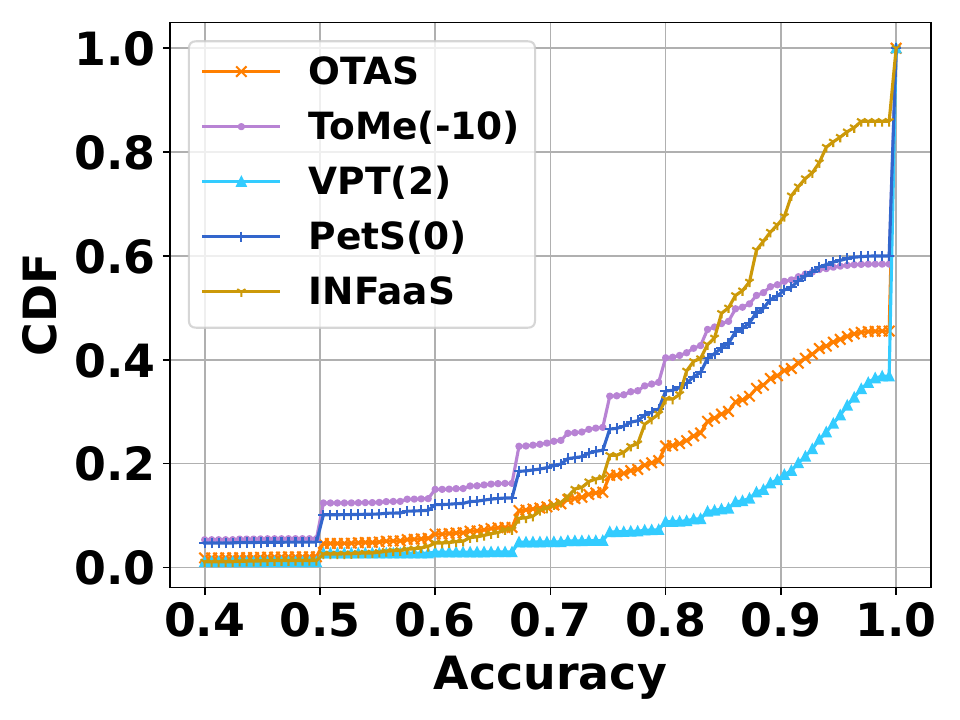}\label{fig:batch_acc_cdf_synthetic}}
	\hfill
	\subfloat[The CDF plot of accuracies on the MAF dataset.]
        {\includegraphics[width = 0.23\textwidth]{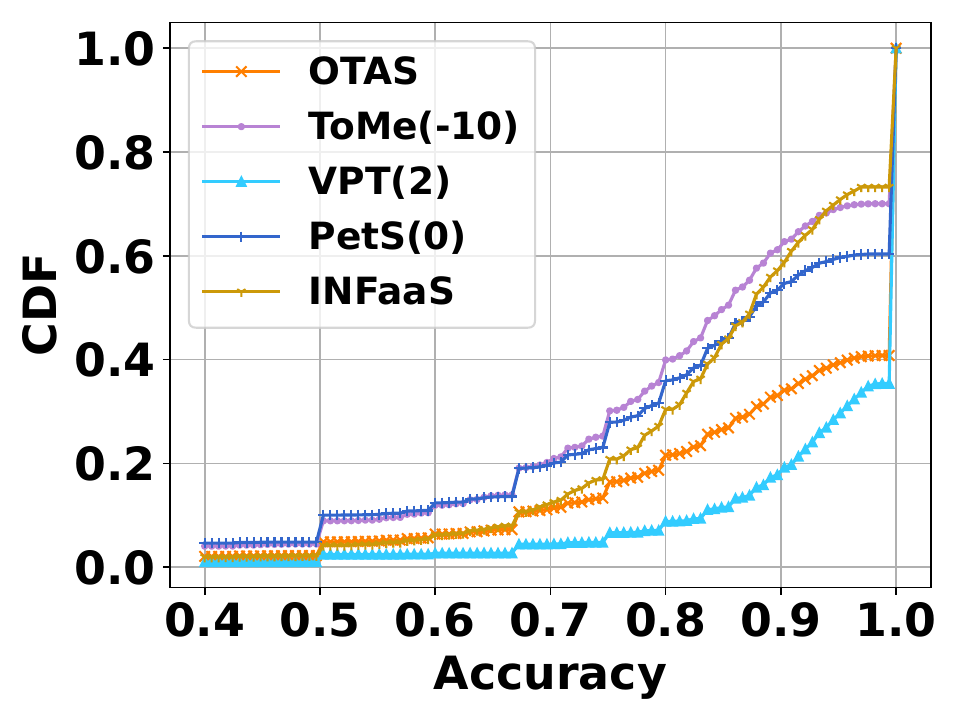}\label{fig:batch_acc_cdf_azure}}
    \caption{The CDF plot of accuracies for served batches.}
    \label{fig:batch_acc_cdf}
    \end{figure}

    \textbf{The accuracies of batches.} We present the CDF plot of accuracies for served batches with five methods on the synthetic dataset in Fig.~\ref{fig:batch_acc_cdf_synthetic}. 
    The VPT method with a prompting number of 2 achieves the highest accuracy because of the incorporation of well-trained prompting tokens. 
    The ToMe method exhibits relatively low accuracy, owing to the reduction of tokens. 
    \textsc{OTAS} can select an appropriate execution scheme dynamically, thereby achieving a balance between accuracy and latency. 
    The average accuracy of our method is larger than $90\%$, indicating that our approach can successfully provide accurate results for served queries. Though INFaaS achieves high accuracy with a stronger model, it comes at the cost of increased I/O overhead.
    As shown in Fig.~\ref{fig:batch_acc_cdf_azure}, the accuracy on the MAF dataset is similar to that observed on the synthetic dataset. 
    The accuracy curve exhibits a sudden increase as it approaches 1, primarily due to the large number of batches with a perfect accuracy score of 1.
    


    \begin{figure}[t]
    \centering
	\subfloat[The $\gamma$ selection of \textsc{OTAS} on the synthetic dataset.]
        {\includegraphics[width = 0.23\textwidth]{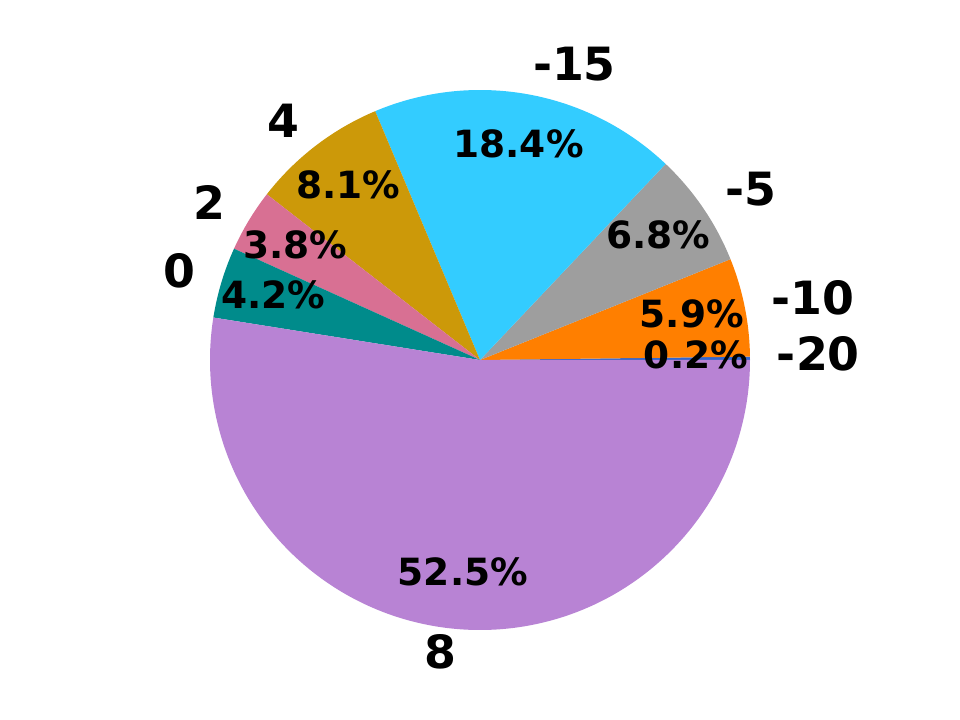}\label{fig:gamma_selection_pie_synthetic}}
	\hfill
	\subfloat[The $\gamma$ selection of \textsc{OTAS} on the MAF dataset.]
        {\includegraphics[width = 0.23\textwidth]{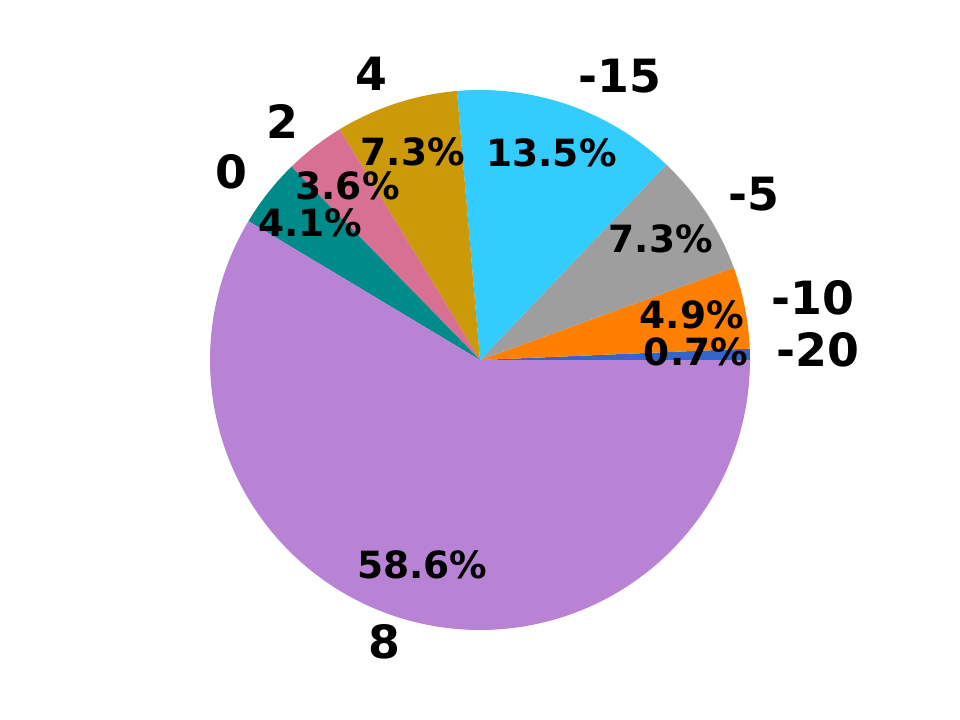}\label{fig:gamma_selection_pie_azure}}
    \caption{The $\gamma$ selection of \textsc{OTAS}.}
    \label{fig:gamma_selection_pie}
    \end{figure}
    

    \textbf{The $\gamma$ selection.} \textsc{OTAS} can change the token number $\gamma$ according to the incoming load and the query characteristics. We present the $\gamma$ selection ratio in Fig.~\ref{fig:gamma_selection_pie}. 
    On the synthetic trace, \textsc{OTAS} selects $\gamma=8$ at most, given the flat request load for most of the serving period.
    Another major selection is $\gamma=-15$ because it can reduce the inference time while keeping the accuracy nearly unchanged.
    On the MAF trace, more batches executed with a prompting number of 8 because of the light query loads. During busy periods, \textsc{OTAS} selects a $\gamma$ value of -15 to serve more queries.

     \begin{figure}[t]
    \centering
	\subfloat[Synthetic dataset.]
        {\includegraphics[width = 0.23\textwidth]{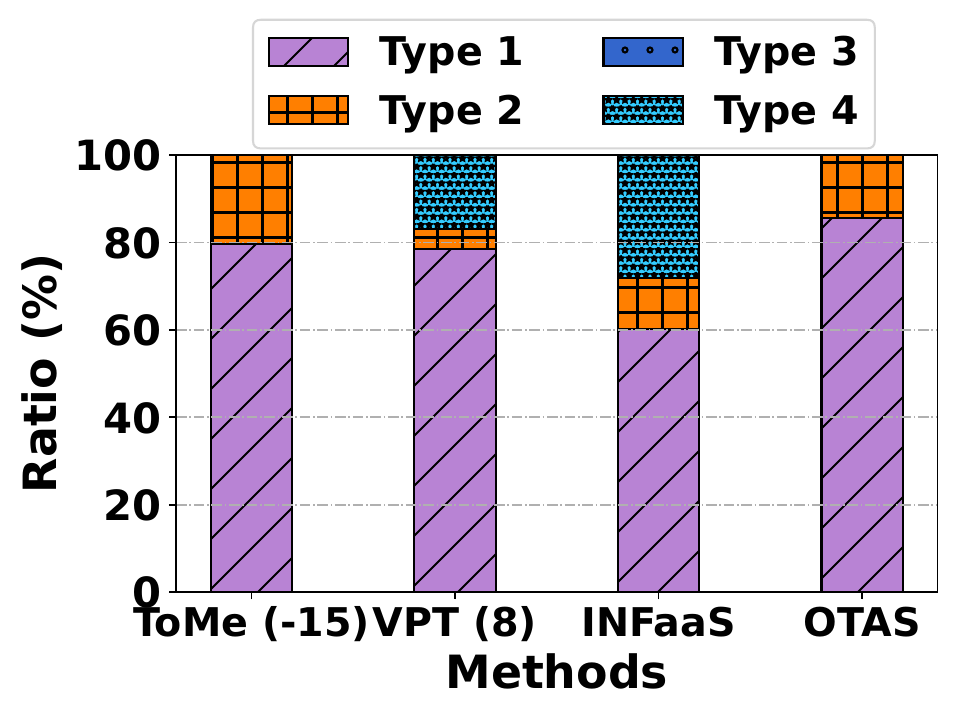}\label{fig:batch_type_synthetic}}
	\hfill
	\subfloat[Azure dataset.]
        {\includegraphics[width = 0.23\textwidth]{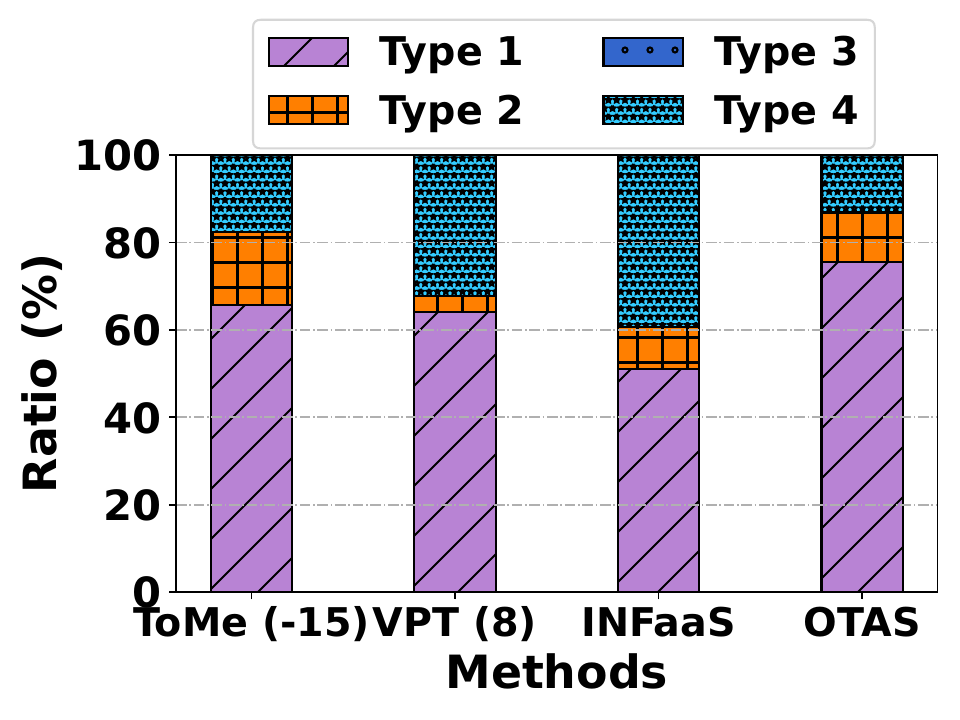}\label{fig:batch_type_azure}}
        \hfill
    \caption{The ratio of execution information of different queries.}
    \label{fig:batch_type}
    \end{figure}
    


    \textbf{The execution type of a query.} Queries have different processing outcomes, which can be classified into the following categories. Type 1 - obtaining accurate results and meeting latency constraints; Type 2 - obtaining incorrect results while still meeting latency constraints; Type 3 - obtaining inference results while unable to meet latency deadlines; and Type 4 - queries that cannot meet latency constraints before actual execution and have been evicted.
    The execution ratio of different query types on the synthetic dataset is visualized in Fig.~\ref{fig:batch_type_synthetic}. It can be observed that \textsc{OTAS} is able to successfully serve 85.54\% of the queries (Type 1), and all queries can meet the latency requirement. 
    On the other hand, ToMe can serve fewer queries because it has a low prediction accuracy. VPT and INFaaS has a longer inference time that leads to a higher eviction ratio. 
    The ratio of different query types on the MAF dataset is presented in Fig.~\ref{fig:batch_type_azure}. 
    Because there are some highly bursty loads in the MAF dataset, the ratio of evicted queries (Type 4) increases due to the limited computational resources. 
    Compared to other methods, \textsc{OTAS} serves the highest number of requests, with a success rate of 75.58\%.
    For the ToMe method, 16.85\% of requests are mispredicted, which still consume computational resources.
    The success ratio of VPT is only 64.14\% due to the high inference latency. 
    Our method is more flexible to deal with the bursty query loads.

\textbf{Discussion}
\textsc{OTAS} can be generalized to different tasks and execution environments: 
(1) Choose a pre-trained transformer model as the foundation model. (2) Investigate the prompt learning method, train the prompt pool, and design the token reduction algorithm. (3) Profile the accuracy and inference latency and determine the $\gamma$ list. (4) Apply the profiling data to Algorithm~\ref{alg_allocate} for adaptively selecting a $\gamma$ value.



\vspace{7pt}
\section{Related Works}

The optimization of serving system is a popular research area in academia and industry.
A number of works designed different optimization schemes to improve serving performance. (1) \emph{Batching}~\cite{crankshaw2017clipper,ali2020batch,cui2022dvabatch,280922,10.1109/INFOCOM48880.2022.9796853}. Grouping the queries together and executing them in a batch can make use of the computational capacity of the hardware and significantly improve the throughput. DVABatch designs a multi-entry multi-exit batching scheme~\cite{cui2022dvabatch}. ORCA proposes an iteration-level batching mechanism for generative models~\cite{280922}. Our batching strategy further considers the service-level characteristics for elastic adaptation.
(2) \emph{Model adaptation}~\cite{10.14778/3570690.3570692,wang2023tabi,10.1145/3575693.3575698,9817634}. To deal with fluctuating query loads, a line of serving systems deploys a hierarchy of models and selects a suitable model dynamically. However, model scaling is infeasible for large transformer models. Firstly, the service provider only deploys a foundation model for all users. Moreover, training different versions of large models consumes numerous computational resources, and switching different models requires high I/O costs. Compressing the model may negatively influence the emerging abilities of the vanilla pre-trained model.   
(3) \emph{Resource management and scheduling}~\cite{fasttransformer,ren2022edgematrix,li2023alpaserve,285173,9796884}. Another class of works focused on optimizing the model parallelism, model placement and resource management to improve the device usage. AlpaServe designs an efficient strategy for placing and parallelizing models~\cite{li2023alpaserve}.

The optimization strategy most relevant to us is model adaptation that achieves elastic serving through switching different model versions. Instead of relying on heavy model selection, we dynamically adjust the execution process through lightweight token adaptation. Besides, batching and dynamic resource provisioning are orthogonal to our method and can be easily incorporated into our framework.



\vspace{7pt}
\section{Conclusion}

We present \textsc{OTAS}, an elastic serving system for large transformers based on an innovative idea of token adaptation. We implement a prototype that supports dynamically allocating the token number for a batch and running the transformer model in flexible way. 
The results show that \textsc{OTAS} can improve the utility by at least 18.2\% on both simulated and real-world production traces from Azure compared with other state-of-the-art methods. The observed performance improvement is achieved because \textsc{OTAS} can identify an optimal balance between the overhead of token increment and the benefits of accuracy improvement based on the real-time query load and user demands. 

\clearpage

\bibliographystyle{IEEEtran}
\bibliography{sample}

\end{document}